\documentclass[10pt]{article}
\usepackage{graphicx}
\usepackage{natbib}
\usepackage{subfigure}
\usepackage[ruled]{algorithm2e}
\usepackage{bbm}
\usepackage{indentfirst,csquotes}

\usepackage{amssymb,amsthm,amsmath}
\usepackage{xcolor,paralist,hyperref,titlesec,fancyhdr,etoolbox}
\newtheorem{theorem}{Theorem}[]

\newtheorem{lemma}[theorem]{Lemma}

\newtheorem{corollary}[theorem]{Corollary}

\hypersetup{ colorlinks=false, linkcolor=black, filecolor=black, urlcolor=black }

\usepackage{lipsum}
\usepackage[labelsep=space]{caption}
\usepackage{subcaption}

\newcommand{\indep}{\perp \!\!\! \perp}

\DeclareMathOperator*{\argmax}{arg\,max}
\DeclareMathOperator*{\argmin}{arg\,min}

\begin{document}
\title{Experimental Designs for Optimizing Last-Mile Delivery} 
\author{Nicholas Rios and Jie Xu \\
        George Mason University}
\date{}
\maketitle


\begin{abstract}
Companies like Amazon and UPS are heavily invested in last-mile delivery problems. Optimizing last-delivery operations not only creates tremendous cost savings for these companies but also generate broader societal and environmental benefits in terms of better delivery service and reduced air pollutants and greenhouse gas emissions. Last-mile delivery is readily formulated as the Travelling Salesman Problem (TSP), where a salesperson must visit several cities and return to the origin with the least cost. A solution to this problem is a Hamiltonian circuit in an undirected graph. Many methods exist for solving the TSP, but they often assume the travel costs are fixed. In practice, travel costs between delivery zones are random quantities, as they are subject to variation from traffic, weather, and other factors. Innovations such as truck-drone last-mile delivery creates even more uncertainties due to scarce data. A Bayesian $D-$optimal experimental design in conjunction with a regression model are proposed to estimate these unknown travel costs, and subsequently search for a highly efficient solution to the TSP. This framework can naturally be extended to incorporate the use of drones and any other emerging technology that has use in last-mile delivery. 

\vspace{2mm}

\end{abstract} 

\section{Introduction}

Last mile delivery has become an indispensable part of modern economy and social fabric and its prevalence and the delivery volume is still growing explosively. Everyday, hundreds of millions of customers received goods and services through last mile delivery. The speed, cost, and reliability of last mile delivery have huge impact on the profit and supply chain efficiency of business entities ranging from juggernauts like UPS and Amazon to small and local business owners, who rely on last mile delivery to reach a broader customer base and more effectively compete with larger companies. Last mile delivery also contributes significantly to local job creations. Besides these benefits, last mile delivery also has significant environmental and climate impact because of emissions of air pollutants and green house gases. There is a growing need for sustainable, environmental friendly, and economically competitive last mile delivery solutions. 

In last mile delivery operations, delivery points are often grouped into several zones. Delivery can be done using a single truck that drives to each zone, and then completes local deliveries in the zone. More recently, companies such as UPS and Amazon have considered a paradigm where trucks deploy unmanned aerial vehicles (drones) to assist in the delivery of packages. Under this paradigm, once a truck arrives in a zone, instead of driving to delivery points in this zone, it deploys drones to deliver packages to all the delivery points in the zone. This will particularly be useful in rural areas, where drones can fly virtually unimpeded and deliver points are often widely dispersed. For example, \cite{park2018comparative} explained that drone delivery is highly efficient in countries with rural areas and low population density; in fact, they found that the reduction of pollution provided by drones was much stronger in rural areas. The medical industry has also found drones useful; \cite{cheskes2020improving} examined the effectiveness of delivering Automated External Defibrillators to victims of cardiac arrest in rural areas via a drone and ambulance. A framework for efficiently delivering test kits and medicine to patients in rural areas has been proposed by \cite{kim2017drone}. 

The problem of interest is, given a cluster of delivery zones, what is the optimal order for the truck to depart from a depot, visit each delivery zone, and return to its point of origin? One solution is to formulate the delivery route as a Hamiltonian circuit on a weighted undirected graph, where the nodes represent the depot and all delivery points. This problem was recently studied by \cite{chang2018optimal}, who used K-means clustering to group the delivery points, and then applied a Travelling Salesman Problem (TSP) solver to find an efficient route. Since the cost of delivery depends on distances between zones and distances between the truck and delivery points, the cost is inherently symmetric. Other variants of this problem have received recent attention. \cite{agatz2018optimization} studied the Traveling Salesman Problem with Drone (TSP-D), where a single truck and a single drone are used. \cite{freitas2022exact} employed a mixed integer programming framework to find efficient heuristic solutions to variants of the TSP-D problem. In practice, it is more time and cost efficient to operate multiple drones from a truck and thus we focus on the problem of finding an optimal route for a single truck that uses multiple drones. 

One key challenge in applying those approaches for real world truck-drone last mile delivery is TSP solvers only work well when there are accurate estimates of travelling times and costs within and between delivery zones. However, in this case, delivery costs and times are be subject to additional uncertainty from weather, birds, and possible location-specific complications for drone visit. To tackle these emerging uncertain factors, there is a great need for new methods to learn the delivery costs and times and optimize delivery zoning and routing in an effective, efficient, and jointly optimal way. 

 In this paper, we propose a Bayesian framework that uses experimental designs to select a small set of Hamiltonian circuits for testing in real time, estimates unknown travel times between delivery zones, and then searches for an efficient route. The cost and time for drone delivery within a zone and the travelling between delivery zones will be treated as random quantities with prior distributions reflecting how much information is available to the delivery company. For example, travel time estimates from the Google Maps API can be used; for examples, see \cite{cahyani2023traveling}, \cite{hameed2020multi}, and \cite{elsayed2020smart}.  These methodological developments will culminate in a readily implementable learning-enabled decision support tool to search for minimal cost delivery zoning and routing in the new truck-drone last mile delivery paradigm.

\section{Problem Formulation}
\label{sec:linmod}

Suppose that there are $m-1$ delivery zones of interest that are fixed in advance. Assume that a truck, which must visit all delivery zones exactly once, begins from a depot, which we denote as zone $m$, for convenience. In this problem, a delivery route is represented by a Hamiltonian circuit in a weighted, undirected graph with $m$ vertices. Denote a delivery route as $\textbf{a} = (a_1, a_2, \dots, a_m, a_1)$, where $(a_1, \dots, a_m)$ is a permutation of $(1,2,\dots,m)$. Let $\mathcal{A}^*$ be the set of all $\frac{(m-1)!}{2}$ distinct Hamiltonian circuits on a complete graph with $m$ vertices, which are invariant under reversals and cyclic permutations. For a delivery route $\textbf{a}$, let $\tau(\textbf{a})$ be its average cost, which includes the cost of returning to the depot. 

We wish to identify an optimal Hamiltonian circuit $\textbf{a} \in \mathcal{A}^*$, i.e. \begin{equation}
\label{eqn:optproblem}
    \textbf{a}^* = \argmin_{\textbf{a} \in \mathcal{A}^*} \tau(\textbf{a})
\end{equation} In many practical applications, $\tau(\textbf{a})$ is unknown, which means that travel costs must first be estimated before applying traditional methods. If the travel costs between different delivery zones were exactly known, then $\tau(\textbf{a})$ could be exactly computed for each Hamiltonian circuit $\textbf{a}$, and efficient (but not exactly optimal) solutions to this problem could be found by the nearest neighbor algorithm \citep{kizilatecs2013nearest} or Christofides algorithm \citep{an2015improving}. However, travel times need to be estimated in the presence of noise and uncertainty in deliveries under a limited experimental budget. In the general truck-drone problem, where one truck deploys multiple drones, additional uncertainty is introduced from other sources as well. For instance, the amount of time a truck has to spend in a single delivery zone depends on the longest delivery trip among all of the drone deliveries made to customers in that particular zone. The drone travel time includes the time it takes to travel to the customer and back to the truck; this is subject to noise depending on wind speeds, customer location relative to the truck, and size and weight of the package. In the proposed framework, experimental designs coupled with statistical models can be used to first estimate these unknown travel costs, and then state-of-the-art algorithms can be used to identify an optimal route.


\section{Proposed Methods}

To solve the optimization problem in (\ref{eqn:optproblem}) in the presence of noise, we propose a Bayesian learning-based optimization framework that integrates statistical models with optimization methods. Let $x_{jk}(\textbf{a}) = 1$ if the undirected edge $(j, k)$ is in the Hamiltonian circuit $\textbf{a}$ and $x_{jk}(\textbf{a}) = 0$ otherwise. Consider Model (\ref{eqn:edgemodel}).\begin{equation}
\label{eqn:edgemodel}
    y(\textbf{a}) = \sum_{j < k}\beta_{jk}x_{jk}(\textbf{a}) + \epsilon(\textbf{a}),  \quad 
    \beta_{jk} \sim N(\mu_{jk}, \sigma_{jk}^2) \indep \epsilon(\textbf{a}) \sim N(0,\sigma^2)
\end{equation} In model (\ref{eqn:edgemodel}),  $\beta_{jk}$ is the stochastic travel cost for the edge $(j,k)$ in the Hamiltonian circuit $\textbf{a}$, $y(\textbf{a})$ is the observed total travel cost for circuit $\textbf{a}$, and $\epsilon(\textbf{a})$ is a noise term for each circuit. We assume that the edge costs are mutually independent and follow a Normal distribution with mean $\mu_{jk}$ and variance $\sigma^2_{jk}$, and that the edge costs are independent of the overall error terms $\epsilon(\textbf{a})$. 

Model (\ref{eqn:edgemodel}) is flexible in the sense that it allows prior information on travel costs to be used in several ways. The means $\mu$ represent our existing (prior) knowledge of the average costs for traveling from delivery zone $j$ to delivery zone $k$. These can simply be distances provided by map routing software, or they can be historical results from existing data. In the case where little prior knowledge of distances is provided, then the entries of $\mu$ can be zero. Similarly, the variances $\sigma_{jk}^2$ represent any prior knowledge of the variation in travel costs corresponding to the edge $(j,k)$. As an example, one method for incorporating traffic information into the prior knowledge is to set $\sigma_{jk}^2 = \lambda_1$ for areas with low traffic fluctuations and $\sigma_{jk}^2 = \lambda_2$ for areas with high traffic fluctuations, where $\lambda_2 > \lambda_1 > 0$. 

Model (\ref{eqn:edgemodel}) is a Bayesian model. The stochastic travel costs $\beta = [\beta_{12}, \dots, \beta_{m-1,m}]^T$ follow a prior distribution $N(\mu, R^{-1})$, where $\mu = [\mu_{12}, \dots, \mu_{m-1,m}]^T$ and $R^{-1}$ is a diagonal matrix with diagonal entries $\sigma^2_{12}, \dots, \sigma^2_{m-1,m}$. In Bayesian analysis, we collect data to update the prior distribution of $\beta$ to get a posterior distribution of $\beta$ given the observed data. When collecting data, a design is required, which is a subset of the Hamiltonian circuits. For a particular design (subset) $D$ of $n$ Hamiltonian circuits, let $X$ be the $n \times {m \choose 2}$ matrix of coded variables $x_{jk}(\textbf{a})$, and let $Y$ be the $n \times 1$ column vector of total costs for each corresponding circuit.  Then, according to standard Bayesian theory, the posterior estimate of $\beta \mid X, Y$ is given by \cite{jones2008bayesian} \begin{align}
    \label{eqn:paramest}
    \hat{\beta} = (X^TX + R)^{-1}(X^TY + R\mu)
\end{align}  It is also known that the distribution of $\beta$ given the data $X, Y$ is Normal with variance proportional to $(X^TX + R)^{-1}$, so it is possible to quantify the uncertainty of the travel costs.

\subsection{Finding Efficient Designs}
Because emerging last-mile delivery solutions such as truck-drone delivery lacks historical data, estimating travel costs would require experimenting with multiple delivery routes. Since there are $(m-1)!/2$ possible  distinct Hamiltonian circuits, it is desirable to identify an efficient design $D^*$ of Hamiltonian circuits to conduct experiments and estimate travel costs. Let $\mathcal{D}$ be the set of all possible subsets of $\mathcal{A}^*$.  Then, an optimal design $D^*$ is defined as \begin{align*}
    D^* = \argmax_{D \in \mathcal{D}} \phi(D)
\end{align*} for some optimality criterion $\phi$. We focus on the Bayesian $D-$optimality criterion \citep{chaloner1982optimal} so that we can minimize the uncertainty in the cost estimates after updating them based on prior information; in turn, this would give us more power for hypothesis testing. This framework also makes the design more robust to bias that would occur from assuming an incorrect model \citep{dumouchel1994simple}. In particular, this criterion maximizes \begin{align}
    \phi(D) = |X_D^TX_D + R|
    \label{eqn:phiD}
\end{align} where $X_D$ is the $X-$matrix obtained from applying the coding in Model (\ref{eqn:edgemodel}) to a subset $D$, which consists of exactly $n$ Hamiltonian circuits. A subset $D$ that maximizes the Bayesian $D-$optimality criterion would minimize the generalized variance of the estimated travel costs (\ref{eqn:paramest}). Minimizing the generalized variance increases the power of statistical tests for detecting significant travel costs. 

\begin{corollary}
\label{cor:fulloptimal}
Let $D_f$ be the set of all $(m-1)!/2$ Hamiltonian circuits. Under model (\ref{eqn:edgemodel}), $D_f$ is $\phi-$optimal for any $\phi$ that is concave and permutation invariant.
\end{corollary}

Corollary \ref{cor:fulloptimal} follows from Theorem 1 of \cite{OOADesign}. The only difference is that in the symmetric case, the optimality criterion must be permutation invariant instead of signed permutation invariant, which was required for the model examined in \cite{OOADesign}. In the case of Bayesian $D-$optimality, this corollary applies because the criterion is also symmetric and permutation invariant. Since Corollary 1 is true, there is an upper bound in optimality of subsets. 

\begin{theorem}
\label{thm:fullmomentmat}
Let $X_f$ be the $X-$ matrix corresponding to $D_f$. Then \begin{align*}
    X_f^TX_f\frac{2}{(m-1)!} = \frac{2}{(m-1)}I + \frac{2}{(m-1)(m-2)}Q
\end{align*} where $I$ is a $p \times p$ identity matrix for $p = {m \choose 2}$. $Q$ is also a $p \times p$ matrix with columns and rows indexed by the pairs $\{(1,2), (1,3), (1,4), \dots, (2,3), (2,4), \dots, (m-1,m)\}$ in lexicographically increasing order, with elements \begin{align*}
    Q(ij, k\ell) = \begin{cases}
        0 \quad \text{ if } ij = k\ell \\
        2 \quad \text{ if } i \neq k, \ell, \text{ and } j \neq k, \ell \\
        1 \quad \text{ otherwise }
    \end{cases}.
\end{align*}
\end{theorem}

Based on Theorem \ref{thm:fullmomentmat}, we can quickly evaluate the Bayesian $D-$optimality criterion for the set of all Hamiltonian circuits, without having to explicitly list all $(m-1)!/2$ circuits. This is because the matrices $I$ and $Q$ have simple closed forms, so it is easy to generate them using standard software and then take the determinant in (\ref{eqn:phiD}). Given the large cardinality of $D_f$, it is of interest to find optimal fractions of $D_f$ to use for experimentation and fitting model (\ref{eqn:edgemodel}). We first show that it is possible to recursively construct $\phi-$optimal half-fractional subsets for arbitrary $m$. This process is summarized in Algorithm \ref{Alg:getoptdesign}.

{\SetAlgoNoLine
\begin{algorithm}[H]
\label{Alg:getoptdesign}
  \textbf{Inputs:} Optimal fractional design for $m-1$ components $D_{m-1}$ with $B$ rows. \\
  \For{$b = 1,2,\dots,B $}{
    1. Let $\textbf{a}_b$ be the $b^{th}$ row of $D_{m-1}$. \\
    2. Form a matrix $C_b$ with $m-1$ rows, where the first row of $C_b$ is $\textbf{a}_b$, and each row is a single clockwise rotation of the previous row. \\
    3. Let $U_b = [\textbf{1}, C_b + 1]$ where $\textbf{1}$ is a vector of ones. \\
  }
  4. $D = [U_1^T, U_2^T, \dots, U_B^T]^T$ \\
  \Return{$D$}

 \caption{ Recursively Generate Optimal Fractional Design }
\end{algorithm}}

Algorithm \ref{Alg:getoptdesign} generates an optimal fractional design for $m$ components given an optimal fractional design for $m-1$ components. Theorem \ref{thm:recursivedesign} shows that Algorithm \ref{Alg:getoptdesign} will produce a $\phi-$optimal design for any optimality criterion $\phi$ that is concave and permutation invariant. The Bayesian $D-$optimality criterion satisfies these conditions, because these are satisfied for $D-$optimality, and the term $R/N$ is the same for all possible subsets of the full design. 

\begin{theorem}
\label{thm:recursivedesign}
Let $D_{m-1}$ be a $\phi-$optimal design for $m-1$ components. Then, if Algorithm \ref{Alg:getoptdesign} is used to construct a design $D_m$, then $D_m$ must also be a $\phi-$optimal design in $m$ components. 
\end{theorem}

To implement Theorem \ref{thm:recursivedesign}, one must first find an optimal design in $m$ components for small $m$, and then repeatedly apply Algorithm \ref{Alg:getoptdesign}. Fortunately, for small $m$, it is possible to use exhaustive search to find optimal design. An optimal design for $m = 5$ is provided in Table \ref{tab:m5opt}.

\begin{table}[!ht]
\centering
\caption{A $\phi-$optimal subset in $m = 5$ components. }
\begin{tabular}{lllll}
\hline \hline 
1 & 2 & 4 & 3 & 5 \\
1 & 2 & 3 & 5 & 4 \\
1 & 2 & 5 & 4 & 3 \\
1 & 5 & 2 & 3 & 4 \\
1 & 3 & 2 & 4 & 5 \\
1 & 3 & 5 & 2 & 4 \\
\hline \hline 
\end{tabular}
\label{tab:m5opt}
\end{table}

\subsection{Algorithmic Construction of Highly Efficient Designs}


In the event that budgets are too heavily constrained, theoretical results may not provide subsets of circuits that are affordable. In such cases, it is possible to use heuristic search algorithms to find a subset that is both affordable and highly efficient in terms of the Bayesian $D-$optimality criteria. Traditionally, these problems have been solved using greedy exchange algorithms, such as the Fedorov algorithm \citep{federovtheory} or the coordinate exchange algorithm \citep{meyer1995coordinate}. These algorithms are ``greedy'' in the sense that at each iteration, they will deterministically choose a new solution that strictly improves the optimality criterion. In general, these exchange algorithms typically converge to a locally (but not necessarily globally) optimal solution \citep{palhazi2016optimal}. The issue with using the Fedorov algorithm is that as $m$ increases, the number of possible Hamiltonian circuits increases dramatically. The Fedorov algorithm requires enumerating all $(m-1)!/2$ possible circuits, which becomes inefficient. 

A great example of a greedy algorithm that can be adapted to this problem is the bubble sorting algorithm proposed by \cite{lin2019order}. This algorithm does not require enumerating all $(m-1)!/2$ possible Hamiltonian circuits. The main idea behind this algorithm is that finding an optimal design can be viewed as a sorting problem, where each row needs to be ``sorted'' in an order that maximizes the optimality criterion. This algorithm takes an initial design with $n$ rows, and (as per the bubble sort algorithm) exchanges adjacent elements in each row of the design if they increase the Bayesian $D-$optimality criterion. More details are given in Algorithm \ref{alg:BB} below.

\begin{algorithm}[H]
\label{alg:BB}
  \SetAlgoLined
  \textbf{Input}: Initial subset $D^{(0)}$ with $n$ rows, and the maximum number of iterations $iter_{max}$.\\
  1. Calculate $d_0 = \phi(D^{(0)})$. \\
  \For{$iter = 1,2,\dots,iter_{max}$}{
     \For{$r = 1,2,\dots,n$}{
        2. \texttt{Still\_Sorting} = True. \\
        \While{\texttt{Still\_Sorting}}{
        \For{$j = i,\dots,m-1$}{
            3. Copy $D^{(1)} = D^{(0)}$. In the $r^{th}$ row of $D^{(1)}$, exchange the elements in positions $j$ and $j+1$. \\
            4. $d_1 = \phi(D^{(1)})$.\\
             \eIf{  $d_1 > d_0$  }{
         5. $D^{(0)} = D^{(1)}$, $d_0 = d_1$.  \\
    }{
    6. \texttt{Still\_Sorting} = False. 
    }
            
        } }
    }
    7. Store the best design so far as $D^*=D^{(0)}$. \\
     }

 \Return{$D^*$}
 \caption{Bubble Sorting Algorithm}
\end{algorithm}

As inputs, Algorithm \ref{alg:BB} takes an initial subset of $n$ rows and a maximum number of iterations. In each iteration, the bubble sorting algorithm examines each row of the current design. For each row, it considers exchanging the positions of adjacent components; if this exchange improves the Bayesian $D-$optimality criterion, the exchange is performed. Once all rows have been examined, the procedure repeats until the maximum number of iterations has been reached. Since this is a greedy local optimization procedure, it is fast, but it is not guaranteed to find the globally best design. It is therefore recommended to run this algorithm for several initial designs and select the best design.

For large spaces, it is helpful to have an algorithm that spends more time exploring the space of possible subsets. An initial example of this is a simulated annealing algorithm \citep{bertsimas1993simulated} to search for highly efficient fractional subsets of arbitrary run size $n \geq 1+{m \choose 2}$. In general, Simulated Annealing works by proposing solutions that are close to the current solution, and accepting them with a probability that depends on the change in optimality between the current and proposed solutions \citep{aarts2005simulated, nikolaev2010simulated, delahaye2019simulated}. Since its introduction, Simulated Annealing  received a lot of attention in the optimization literature, with applications in neural network optimization \citep{he2022bearing}, electronic waste disassembly \citep{wang2021genetic}, feature selection in medical datasets \citep{abdel2021hybrid}, vehicle routing \citep{ilhan2021improved} and job shop scheduling \citep{fontes2023hybrid}. The pseudo-code for Simulated Annealing is presented in Algorithm \ref{Alg:SA}.

\begin{algorithm}[H]
\label{Alg:SA}
  \SetAlgoLined
  \textbf{Inputs}: Initial subset $D^{(0)}$ with $n$ rows and model expansion $X^{(0)}$, number of iterations $n_i$. \\
  1. Let $M^{(0)} = (X^{(0)})^TX^{(0)}$. Let $t = 0$.
  2. Store the current $\phi-$efficiency of $M^{(0)}$. \\
  \For{$t = 0,1,2,\dots,n_i$}{
    3. Randomly select a row $i$ and column $j$ of $D^{(t)}$ over $i = 1,\dots,n$ and $j = 2,\dots,m-1$.  \\
    4. Copy $D^{(t+1)} = D^{(t)}$, and exchange elements in positions $j$ and $j+1$ in the $i^{th}$ row of $D^{(t+1)}$.  \\
    5. Find $M^{(t+1)}$ and update the $\phi-$efficiency of $M^{(t+1)}$. Simulate $U \sim U[0,1]$. \\
    \If{  $U < \exp{ \Big( ( - (\phi(M^{(t+1)}) - \phi(M^{(t)})  )/(1/\log(t+1)) \Big)}$  }{
        6. $D^{(t)} = D^{(t+1)}$, $M^{(t)} = M^{(t+1)}$. \\
        7. Store the best subset so far as $D^*$.
    }
  }
  
 \Return{$D^*$}
 \caption{Find $\phi-$Efficient Subset of Hamiltonian Circuits}
\end{algorithm}

As inputs, Algorithm \ref{Alg:SA} takes an initial subset $D^{(0)}$ which must be of the target run size $n$. In Steps 1 and 2, the $\phi-$efficiency of the initial subset is found and stored. During each iteration of Algorithm \ref{Alg:SA}, a neighboring subset is generated by selecting a random row from the current subset (Step 3), and then randomly exchanging two adjacent elements in columns $2$ to $m$ (Step 4). This ensures that the proposed subset is similar to the current subset, and that the new row still corresponds to a Hamiltonian circuit. In Step 5, the $\phi-$efficiency of the proposed subset is found. In Steps 6-7, the proposed subset is accepted with probability $\exp{ \Big( ( - (\phi(M^{(t+1)}) - \phi(M^{(t)})  )/(1/\log(t+1)) \Big)}$. This way, if a proposed subset has slightly lower $\phi-$efficiency than the current subset, it will have a higher chance of being accepted earlier on in the algorithm (when $t$ is lower). This prevents the algorithm for becoming prematurely ``stuck'' at local optimal subsets. As $t$ increases, the probability of choosing subsets with lower $\phi-$efficiency decreases. As a metaheuristic algorithm, simulated annealing is not guaranteed to find the optimal subset in all cases. \\

\subsection{Identification of Optimal Order}

Once a model of the form (\ref{eqn:edgemodel}) is fit to the experimental data, it remains to use it to predict the optimal order. The minimum-cost Hamiltonian circuit could be found by using the model to predict the total cost for all possible Hamiltonian circuits in the graph, and then selecting the circuit that minimizes the predicted total cost, making it so not every single route needs to be run in realtime in practice. However, for very large $m$, it is still inefficient to find all $(m-1)!/2$ predicted total costs. 

An alternative solution is to use a heuristic algorithm on the regression coefficients. one example of this is provided in Algorithm \ref{Alg:getoptpath} below, which uses a Nearest Neighbor (NN) heuristic to search for low-cost delivery routes. In general, it is possible to run any commonly heuristic with the estimated edge costs obtained from the regression coefficients; see Section 4.2 for an example.

\begin{algorithm}[H]
  \textbf{Inputs:} Estimated Vector of Coefficients $\hat{\beta}$ from Model (\ref{eqn:edgemodel}) \\
  1. Initialize a path $L = [m]$, and let $V = \{1,\dots,m-1\} $. \\
  2. Select a random vertex in $V$, and append it to $L$. \\
  \While{ $V \neq \emptyset $}{
    3. Let $v$ be the last element added to $L$. \\
    4. Let $u = \argmin_{u \in V} \hat{\beta}_{uv} $ (or $\hat{\beta}_{vu}$, if $v < u$). \\
    5. Remove $u$ from $V$, and append $u$ to $L$.
  
  }

  \Return{$\hat{\textbf{a}} = L$}

 \caption{ Search for Optimal Hamiltonian circuits }
 \label{Alg:getoptpath}
\end{algorithm}

Algorithm \ref{Alg:getoptpath} uses nearest neighbors to search for an efficient Hamiltonian circuit. In Step 1, a path is initialized with the starting vertex as $m$. In Step 2, a random vertex in $\{1,2,\dots,m-1\}$ is appended to the path $L$. In Steps 3 to 5, the edge incident to the vertex added last that has the lowest cost (measured by the regression coefficients) is appended to the path. This process repeats until all vertices have been visited. The vertex $m$ is used as the initial point because Model (\ref{eqn:edgemodel}) is constrained so that the edge cost for all edges incident to vertex $m$ is zero. For this reason, it is helpful to execute Algorithm \ref{Alg:getoptpath} with several different random vertices initially chosen after vertex $m$. 

Algorithm \ref{Alg:getoptpath} is a heuristic algorithm and, like most heuristic algorithms for the TSP, only provides approximate or highly efficient solutions, and is not guaranteed to find the optimal solution in every case. However, it is possible to provide an upper bound on the efficiency of the solution provided by Algorithm \ref{Alg:getoptpath}. This is addressed in Theorem \ref{thm:minproblimit}. 

\begin{theorem}
\label{thm:minproblimit}
Let $\textbf{a}^* = \argmin_{\textbf{a} \in D_f} \tau(\textbf{a})$, where $\tau(\textbf{a}) = \sum_{jk} \beta_{jk}^*x_{jk}(\textbf{a})$ and $\beta_{jk}^*$ is the true travel cost from delivery zone $j$ to $k$. Let the output of Algorithm \ref{Alg:getoptpath} be $\hat{\textbf{a}}$. Let $\hat{y}(\textbf{a})$ be the total travel cost for circuit $\textbf{a}$ obtained from fitting Model (\ref{eqn:edgemodel}) to the data $(\textbf{a}_i, y_i), i = 1,\dots,n$. Assume that \begin{enumerate}
    \item[(a)] \textbf{Consistency.} $X^TX/n \rightarrow Q$ for some positive definite $Q$ as $n \rightarrow \infty$, and $(1/n)X^T\boldsymbol\epsilon \rightarrow \textbf{0}$ as $n \rightarrow \infty$.
    \item[(b)] \textbf{Triangle Inequality.} For any three vertices $i,j,k$, it must be true that $\beta_{ij}^* + \beta_{ik}^* \geq \beta_{jk}^*$.
\end{enumerate}

Then, \begin{align*}
    P(\hat{y}(\hat{\textbf{a}}) \leq C\tau(\textbf{a}^*)) \rightarrow 1,
\end{align*} as $n \rightarrow \infty$, where $C = (0.5\lceil\log_2(m)\rceil + 0.5)$. 
\end{theorem}

Theorem \ref{thm:minproblimit} shows that, as the sample size increases, the solution produced by Algorithm \ref{Alg:getoptpath} will likely be at most $(0.5\lceil\log_2(m)\rceil + 0.5)\tau(\textbf{a}^*)$, where $\tau(\textbf{a}^*)$ is the cost of the true optimal Hamiltonian circuit. Assumption (a) of Theorem \ref{thm:minproblimit} is required for convergence of the estimator in Equation (\ref{eqn:paramest}). Assumption (b) is required to guarantee the upper bound on the nearest neighbor approximation. Since costs in the symmetric TSP are inherently based on distances or travel times, it is reasonable to require Assumption (b) in many cases.

\section{Results}

In this section, simulations are used to evaluate the efficacy of the algorithmic design construction methods proposed in Section 3, in addition to the quality of solutions found for the TSP in a last mile-delivery setting. 

\subsection{Relative Efficiency Comparison}

It was of interest to compare the performance of Algorithm \ref{alg:BB} and Algorithm \ref{Alg:SA}. The metric for comparison was the relative $D-$efficiency. For a design $D$ with $n$ rows, the relative $D-$efficiency to the full design $D_f$ (with all $(m-1)!/2$ possible Hamiltonian circuits)  is \begin{align}
    \Bigg(\frac{ |X_D^TX_D/n + R| }{ |X_f^TX_f/N + R|}\Bigg)^{1/{m \choose 2}},
\end{align} where $X_D$ and $X_f$ are the model matrix expansions of $D$ and $D_f$, respectively.  A relative $D-$efficiency of 1 indicates that the smaller design $D$ has the same Bayesian $D-$efficiency as the full design, and is therefore optimal. Since the full design is optimal, the relative efficiency will always be between 0 and 1. 

Both algorithms were used to search for optimal designs for $m = 6,8,10$ and $n = {m \choose 2} + 1, 2{m \choose 2} + 1$, and $3{m \choose 2} + 1$. These sample sizes were chosen to illustrate the effect of increasing the run size, while also accounting for the fact that if the number of delivery zones $m$ increases, more experimental runs are needed to provide stable estimates of the model parameters. We assumed that $R = \tau^2I$ with $\tau^2 = 0.01$. Since both algorithms are sensitive to the initial choice of design, each algorithm was executed for 10 different initial designs, and the best design in each iteration was selected. For each combination of $m$ and $n$, both algorithms were executed 100 times.

\begin{table}[!ht]
\centering
\caption{Comparison of Median Relative $D-$efficiencies}
\begin{tabular}{ccccc}
\hline \hline 
$m$  & $n$   & $n/N$   & Simulated Annealing & Bubble Sort \\
\hline 
6  & 16  & 0.263 & 0.961               & 0.963       \\
   & 31  & 0.516 & 0.978               & 0.996       \\
   & 46  & 0.767 & 0.984               & 0.997       \\
8  & 29  & 0.011 & 0.880               & 0.919       \\
   & 57  & 0.023 & 0.931               & 0.981       \\
   & 85  & 0.034 & 0.952               & 0.992       \\
10 & 46  & $<$0.001 & 0.807               & 0.875       \\
   & 91  & $<$0.001 & 0.898               & 0.967       \\
   & 136 & $<$0.001 & 0.929               & 0.985      \\
   \hline \hline 
\end{tabular}
\label{tab:reldeffcomp}
\end{table}

Table \ref{tab:reldeffcomp} shows the median relative $D-$efficiencies of 100 designs found by the Simulated Annealing and Bubble Sort algorithms for $m = 6,8,10$. As expected, if the sample size $n$ increases, the median relative $D-$efficiency of both methods increases. The third column of Table \ref{tab:reldeffcomp} shows the ratio of the sample size $n$ to the number of possible Hamiltonian circuits $N = (m-1)!/2$. In all cases, it is clear that Bubble Sort achieves a higher median relative $D-$efficiency than Simulated Annealing, and it is the preferred algorithm. Even when fewer than $0.1\%$ of the total available Hamiltonian circuits were included in the design, the Bubble Sort algorithm is able to identify designs with a relative $D-$efficiency of 98\%. 

\subsection{Efficiency of Estimated Optimal Routes}

To demonstrate the efficacy of the proposed methods, we compared the true route costs of the solutions identified by Algorithm \ref{Alg:getoptpath} using three different heuristics to solve the TSP. These three heuristics are nearest neighbors (NN), arbitrary insertion (ARB), and Concorde's Chained Link-Kernighan (LK) heruristic \citep{applegate2003chained}. These three algorithms were implemented in R using the package \texttt{TSP} \citep{hahsler2014tsp}. The goal of this simulation was to examine the distribution of route costs (for all three heuristics) using three different methods of estimating the edge costs. The first estimation method is a naive approach that only uses given prior distances between delivery zones. The second estimation method is a frequentist method, where ridge regression was used to estimate the coefficients in a model of the form \begin{align}
    \label{eqn:ridgemodel}
    y(\textbf{a}) = \sum_{j < k}\beta_{jk}x_{jk}(\textbf{a}) + \epsilon(\textbf{a}).
\end{align} Notice that Model (\ref{eqn:ridgemodel}) is a frequentist regression model, as it does not assume a prior distribution on the parameters $\beta_{jk}$. Ridge regression is used because this model is over-parameterized. For any Hamiltonian circuit $\textbf{a}$ and for each vertex $i \in \{1,\dots,m-1\}$, $\sum_{jk \ni i}x_{jk}(\textbf{a}) = 2$; i.e., each vertex is always adjacent to two other vertices. It follows that, without any regularization or identifiability constraints, the resulting model matrix $X$ will not be full rank. Ridge regression was implemnted using the R package \texttt{glmnet} \citep{glmnet2010}, with the penalty parameter $\lambda$ selected via 10-fold cross validation.  The third method is the proposed Bayesian model (\ref{eqn:edgemodel}), which assumes $\beta \sim N(\mu, R)$ with $R = \tau^2I$ for $\tau = 0.1$. 


In this simulation, we examined $m$ delivery locations for $m = 20$ that were randomly placed in the region $[0,1] \times [0,1]$. The distance matrix for edges between every pair of locations was found, and each undirected edge in the complete graph on $m$ vertices was weighted according to its distance. Note that the total cost of a Hamiltonian circuit included the return cost. For each value of $m$, experimental designs of size $n = 49, 96, 191,$ and $381$ were found using a Simulated Annealing algorithm with 10,000 iterations. These sample sizes correspond to $n = 0.25{m \choose 2} + 1, 0.5{m \choose 2} + 1, {m \choose 2} + 1$, and $2{m \choose 2}+1$, respectively. These were chosen to show the effect of increasing the sample size on the results. A sample size of $n = {m \choose 2} + 1$ would be the bare minimum needed to estimate the ${m \choose 2}$ main effects in a regression model without regularization or a prior distribution.  

100 iterations of the simulation were executed. During each iteration, random noise was added to the edge weights to simulate fluctuations in travel times. Two random noise scenarios were considered. In Scenario (a), edges with weights below the $25^{th}$ percentile of all edge weights were identified, and $N(0.5, 0.25^2)$ errors were added to these edge weights. In Scenario (b), $t_3(1/\mu_{jk})$ random errors were added to each edge, where $1/\mu_{jk}$ is a non-centrality parameter. In both scenarios, the random noise simulated increased travel times on congested shorter roads, which will potentially have an effect on the optimal route. Finally, $N(0, 0.1^2)$ white noise was added to all edges in both scenarios. The modified edge weights were used to simulate the travel time for each Hamiltonian circuit in the subset. The simulated data were used to fit Model (\ref{eqn:edgemodel}) and Model (\ref{eqn:ridgemodel}) to obtain Bayesian and frequentist estimates of the true edge costs, respectively. Finally, each TSP heuristic was used to search for an optimal route using the naive prior costs, the frequentist estimated costs, and the Bayesian estimated costs. The cost of these optimal routes were stored for comparison.

\begin{figure}
  \centering
  \subfigure[$n = 49$]{\includegraphics[scale = 0.6]{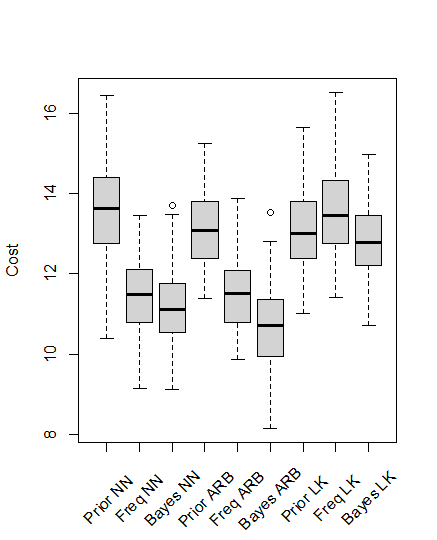}\label{fig:Scenarioam20n49}}
  \subfigure[$n = 96$]{\includegraphics[scale = 0.6]{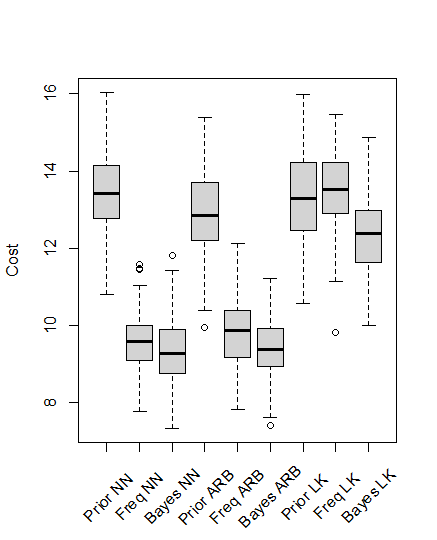}\label{fig:Scenarioam20n96}}
  \subfigure[$n = 191$]{\includegraphics[scale = 0.6]{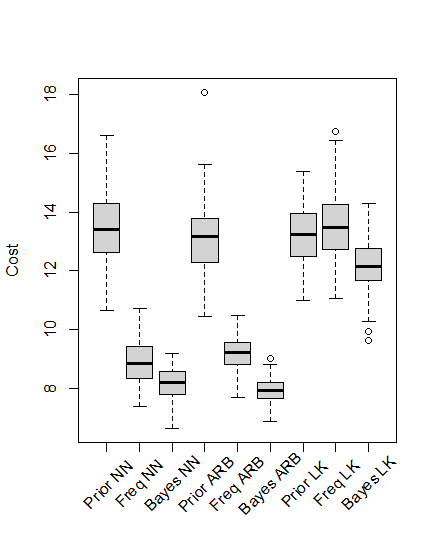}\label{fig:Scenarioam20n191}}
  \subfigure[$n = 381$]{\includegraphics[scale = 0.6]{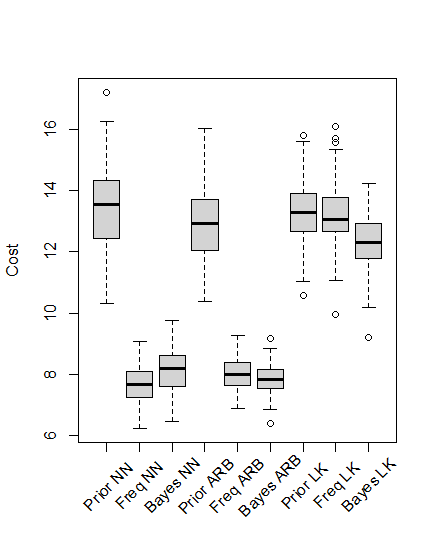}\label{fig:Scenarioam20n381}}
  \caption{Solution Cost Comparison, $m = 20$ Delivery Zones, Scenario (a)}
  \label{fig:CostScenarioa}
\end{figure}

Figure \ref{fig:CostScenarioa} shows the distributions of the true costs of the proposed solutions for each combination of the three heuristics (NN, ARB, and LK) and the three edge cost estimation methods (prior only, frequentist, and Bayes) under Scenario (a) for four different sample sizes. For the NN and ARB heuristics, it is always better to use either the frequentist or Bayesian methods instead of just the prior costs, since the Bayesian and frequentist methods have lower median costs for $n = 49, 96, 191,$ and $381$. For the LK heuristic, it is generally best to use the Bayesian method. The frequentist method only shows improvement over the prior costs when the sample size is very large, e.g., when $n = 381$. For each examined sample size, the LK heuristic had the worst performance in terms of median route cost for the Bayesian and frequentist methods. For smaller sample sizes (e.g. $n = 49, 96$), it is better to use a Bayesian estimate for the costs with the ARB or NN heuristic, as they have similar median costs that are lower than those produced by the other methods. For very large sample sizes, e.g. $n = 381$, the frequentist method with the NN heuristic had marginally better performance than the Bayesian method with ARB and NN. This is intuitive, as when $n$ becomes very large, the likelihood function provides more information relative to the prior distribution.

\begin{figure}
  \centering
  \subfigure[$n = 49$]{\includegraphics[scale = 0.6]{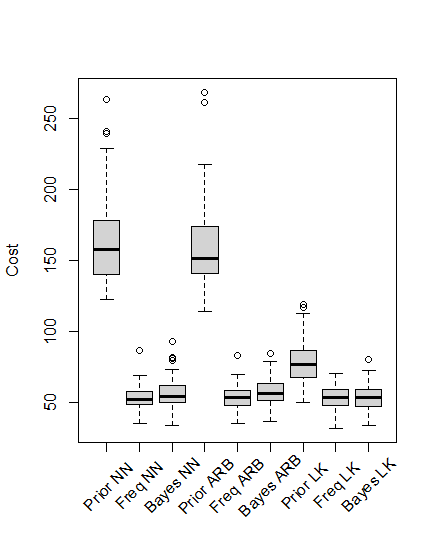}\label{fig:Scenariobm20n49}}
  \subfigure[$n = 96$]{\includegraphics[scale = 0.6]{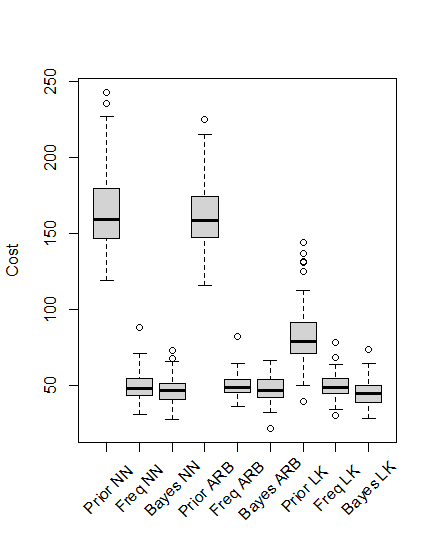}\label{fig:Scenariobm20n96}}
  \subfigure[$n = 191$]{\includegraphics[scale = 0.6]{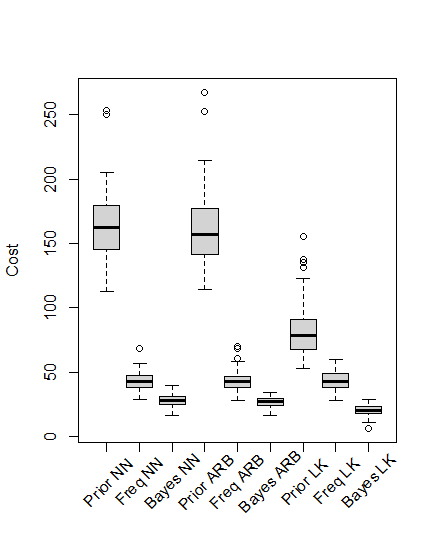}\label{fig:Scenariobm20n191}}
  \subfigure[$n = 381$]{\includegraphics[scale = 0.6]{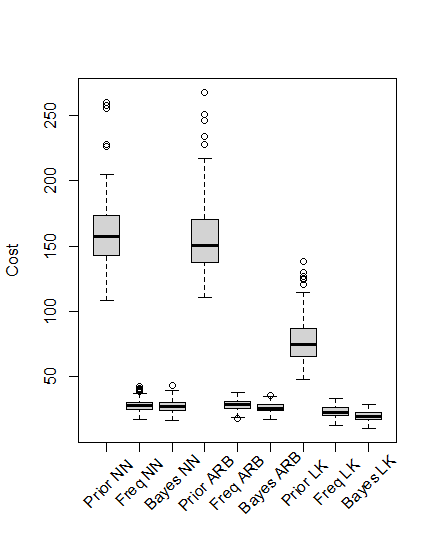}\label{fig:Scenariobm20n381}}
  \caption{Solution Cost Comparison, $m = 20$ Delivery Zones, Scenario (b)}
  \label{fig:CostScenariob}
\end{figure}

Figure \ref{fig:CostScenariob} shows the distributions of the true costs of the proposed solutions for each combination of the three heuristics (NN, ARB, and LK) and the three edge cost estimation methods (Prior, Frequentist, and Bayes) under Scenario (b). Unlike Scenario (a), the LK heuristic has the best performance for $n = 96, 191,$ and $381$ for the prior, frequentist, and Bayesian estimates. In this scenario, LK does substantially better at finding low-cost routes than NN and ARB with only the prior information. All three heuristics see improved results (in terms of lower median cost) when either frequentist or Bayes estimates are used instead of prior only costs. As the sample size increases, there is a larger benefit to using these methods. When $n = 381$, the frequentist and Bayesian methods have very similar performance for the NN and ARB statistics, but when $n = 96, 191$, the Bayesian methods have slightly lower median route costs.

\section{Conclusion}

This paper provides a novel framework for finding low-cost routes for last mile delivery problems via Bayesian D-optimal experimental designs. Theory proved that the full design (with all $(m-1)!/2$ Hamiltonian circuits used exactly once) is $\phi-$ optimal for any criterion $\phi$ that is concave and permutation invariant. A recursive algorithm was developed to construct $\phi-$optimal fractional designs for any number of delivery zones. In the event where these designs were too large, Simulated annealing and bubble sort algorithms were proposed for finding highly efficient designs (subsets of Hamiltonian circuits) to examine. 

In Section 4.2, it was shown that using data-driven approaches, whether frequentist or Bayesian, generally lead to more efficient solutions to the TSP. The results show that the proposed Bayesian estimation paradigm consistently yields low-cost delivery routes when coupled with common heuristics used to solve the TSP.  This is true even for smaller sample sizes, where the Bayesian methods generally found cheaper routes than their frequentist counterparts. The normal Bayesian prior used to fit the proposed model and design the proposed experiment performed well even in a scenario where the true prior was different from the assumed prior. 

In general, finding the shortest Hamiltonian circuits in a graph has many possible applications. In computer science and networking, this methodology would be helpful for finding a route to send packets to $m$ nodes with the lowest total ping. In electrical engineering and city design, this would be helpful for finding a way to arrange power lines between $m$ stations with minimal cost. Furthermore, if smaller optimal fractions are found, then it would be a natural step towards providing efficient solutions for the TSP for larger networks. In its current state, the designs in this paper are also useful for providing a set of initial points to check in an optimization algorithm for the TSP. 

There are several ways this research can be expanded. The model in this paper only examined the effects of edges between vertices in a graph. More complex models could also consider interaction effects between edges in a graph, but there are many possible interactions to consider. Additionally, it may also be of interest to not only minimize the average travel cost, but also the variance of the travel costs. This dual-response optimization problem would require careful attention. Furthermore, it would also be interesting to see how one might search for the optimal route in an online setting where information on route costs comes to a data analyst in batches.

\bibliographystyle{apalike}
\bibliography{biblio}

\appendix

\newpage
\section{Appendix - Proofs}

\begin{lemma}
\label{lemma:permmatrix}
Consider the coding scheme used in model (1). Let $\pi$ be a permutation of $(2,\dots,m)$. For a Hamiltonian circuit $\textbf{a} \in \mathcal{A}$, let $\pi\textbf{a} = (1,a_{\pi_2}, a_{\pi_3} \dots,a_{\pi_m}, 1)$. Then there is a permutation matrix $R_{\pi}$ such that $x(\pi \textbf{a})^T = x(\textbf{a})^TR_{\pi}$.
\end{lemma}

\begin{proof}
Let $\pi \in \mathcal{A}$ and $\textbf{a} \in \mathcal{A}$. If $i$ and $j$ are adjacent in $\textbf{a}$, then by definition, $\pi_{\textbf{a}_i}$ must be adjacent to $\pi_{\textbf{a}_j}$. Hence, it follows that $x_{ij}(\textbf{a}) = x_{\pi_i,\pi_j}(\pi\textbf{a})$ for any pair $ij$. Let $R_{\pi}$ be a matrix with columns and rows indexed by the pairs $(12,13,...(m-1)m)$ such that $R_{\pi}(ij,\pi_i,\pi_j) = 1$ for each pair $ij$ and the remaining elements are 0. Then $x(\pi \textbf{a})^T = x(\textbf{a})^TR_{\pi}$ for any $\pi \in \mathcal{A}$. 
\end{proof}

\noindent \textbf{Proof of Corollary \ref{cor:fulloptimal}.}

\begin{proof}
    The structure of this proof is very similar to that of Theorem 1 in \cite{OOADesign} (which is why we denote this as a Corollary). We write the proof in its entirety for clarity and convenience.     

    Let $B = \frac{(m-2)!}{2}$, and let $\mathcal{P}_{2:m} $ denote the set of all $\frac{(m-1)!}{2}$ distinct permutations of the elements $\{2,3,\dots,m\}$ that are invariant under reversals. Notice that $\mathcal{A} = \cup_{b=1}^B \{ (1,\textbf{a}',1) \mid \textbf{a}' \in \mathcal{P}_{2:m}   \}$. This implies that the matrix of all circuits $D_f$ has the general form $D_f = [\mathbf{1},  D_{2:m}]$, where $\mathbf{1}$ is a column vector of ones. Let    $M_f = \sum_{\textbf{a} \in \mathcal{A}} 2/(m-1)! x(\textbf{a})x(\textbf{a})^T = 2/(m-1)! X_f^TX_f$ be the moment matrix for the uniform design measure, and let $w$ be an arbitrary design measure over $\mathcal{A}$. Let $\pi\textbf{a}$ be a Hamiltonian circuit $(1, a_{\pi_2}, a_{\pi_3}, \dots, a_{\pi_m}, 1)$, where $(\pi_2, \dots, \pi_m) \in \mathcal{P}_{2:m}$, and let $\pi w$ be the design measure that assigns, for each $\textbf{a} \in \mathcal{A}$, weight $w(\textbf{a})$ to the circuit $\pi \textbf{a}$. By concavity of $\phi$, we have that \begin{align}
    \label{ineq:concave}
    \phi\Big( \sum_{\pi \in \mathcal{P}_{2:m}} \frac{2}{(m-1)!} M(\pi w) \Big) \geq \sum_{\pi \in \mathcal{P}_{2:m} } \frac{2}{(m-1)!} \phi\Big( M(\pi w) \Big) 
    \end{align} Notice that, for a fixed $\textbf{a} \in \mathcal{A}$, $\{ \pi\textbf{a} \mid \pi \in \mathcal{P}_{2:m}  \} = \mathcal{A},$ which implies \begin{align}
     &\frac{2}{(m-1)!} \sum_{\pi \in \mathcal{P}_{2:m}} M(\pi w) =  \frac{2}{(m-1)!}\sum_{\textbf{a} \in \mathcal{A}} \sum_{\pi \in \mathcal{P}_{2:m}} w(\pi \textbf{a}) x(\pi\textbf{a})x(\pi\textbf{a})^T  \\
      \label{eqn:LHS}
     = & \frac{2}{(m-1)!} \sum_{\textbf{a} \in \mathcal{A}} w(\textbf{a}) \frac{(m-1)!}{2}  M_f = M_f.
    \end{align} Also notice that, by Lemma 1, \begin{align}
    &\phi\Big( M(w)  \Big) = \phi\Big( \sum_{\pi \in \mathcal{P}_{2:m}} w(\pi \textbf{a}) x(\pi\textbf{a})x(\pi\textbf{a})^T \Big) = \\
    \label{eqn:RHS}
    &\phi(R_{\pi}^TM(w)R_{\pi}) = \phi(M(w))
    \end{align} where the last equality follows because $\phi$ is permutation-invariant. By substituting (\ref{eqn:LHS}) and (\ref{eqn:RHS}) into the inequality (\ref{ineq:concave}), we find that \begin{align}
        \phi\Big(  M_f \Big) \geq 
 \frac{2}{(m-1)!}\sum_{\pi \in \mathcal{P}_{2:m} }  \phi\Big( M(w) \Big)  = \phi\Big( M(w) \Big),
    \end{align} for any design measure $w$ over $\mathcal{A}$. This concludes the proof.

\end{proof}

\noindent \textbf{Proof of Theorem \ref{thm:fullmomentmat}}. 
\begin{proof}
Let $m$ be the number of components and $p = {m \choose 2}$. Let $G = (V,E)$ be a complete, undirected graph with vertices $ V = \{1,2,\dots,m\}$. Let $D_f$ be the design which includes all Hamiltonian circuits in $G$ exactly once, and let $X$ be the model matrix expansion of $D_f$ under the symmetric OofA model (1). Arrange the columns of $X$ by the edges $(i,j)$ of $G$, with $i < j$, in lexicographically increasing order, i.e. $(1,2), (1,3), \dots, (2,3), (2,4), \dots (m-1,m)$. Denote the column of $X$ that corresponds to $(i,j)$ as $X_{ij}$. 

Then, by combinatorial arguments, the following must be true: \begin{enumerate}
    \item[(a)] $X_{ij}^TX_{ij} = \sum_{\textbf{a} \in \mathcal{A}} x_{ij}(\textbf{a}) = (m-2)!$, which is the number of times the undirected edge $(i,j)$ appears in $D_f$. This is true because if the edge $(i,j)$ is fixed, there are $(m-2)!$ permutations of the other vertices, each of which corresponds to a unique Hamiltonian circuit in $G$. Since the set $\mathcal{A}$ is invariant under reversals, we do not count reversals of the pair, i.e. we do not count any Hamiltonian circuit formed with $( j, i)$, as it is already counted.
    \item[(b)] $X_{ij}^TX_{j\ell} = \sum_{\textbf{a} \in \mathcal{A}} x_{ij}(\textbf{a})x_{j\ell}(\textbf{a}) = (m-3)!$. To show this, first fix the position of the triple $(i,j,\ell)$. There are $(m-3)!$ possible permutations of the remaining vertices, each of which corresponds to a distinct Hamiltonian circuit in $G$. Since the set $\mathcal{A}$ is invariant under reversals, we do not count reversals of the triple, i.e. we do not count any Hamiltonian circuit formed with $(\ell, j, i)$, as it is already counted.  Similarly, $X_{ij}^TX_{kj} = (m-3)!$.
    \item[(c)] $X_{ij}^TX_{k\ell} =  \sum_{\textbf{a} \in \mathcal{A}} x_{ij}(\textbf{a})x_{k\ell}(\textbf{a}) = 2(m-3)!$, for $i \neq k,\ell$ and $j \neq k, \ell$. To show this, fix the positions of $(i,j)$ and $(k,\ell)$. Suppose the edge $(i,j)$ is placed before $(k,\ell)$. Then, the rest of the Hamiltonian circuit is specified by $(m-2)$ edges; however, once $(m-3)$ edges are selected, the last edge is fixed (as the circuit must return to the starting vertex); therefore, there are $(m-3)!$ possible arrangements of these edges and $(m-3)!$ corresponding Hamiltonian circuits. By symmetry, if $(k,\ell)$ is placed before $(i,j)$, there are also $(m-3)!$ possible Hamiltonian circuits. Since these cases are disjoint, then there are a total of $2(m-3)!$ Hamiltonian circuits in this case.
\end{enumerate}

Let $N = \frac{(m-1)!}{2}$, which is the number of distinct Hamiltonian circuits on $G$. From (a), (b), and (c), we have that \begin{align*}
    M_f &= \frac{1}{N}X^TX = \frac{2}{(m-1)!}\begin{bmatrix}
        X_{12}^TX_{12}  & X_{12}^TX_{13} & \dots & X_{12}^TX_{m-1,m} \\
        X_{13}^TX_{12}  & X_{13}^TX_{13} & \dots & X_{13}^TX_{m-1,m} \\
        \vdots   &  \vdots  & \ddots   & \vdots \\
        X_{m-1,m}^TX_{12}  & X_{m-1,m}^TX_{13}  & \dots  & X_{m-1,m}^TX_{m-1,m}
    \end{bmatrix} \\
    &= \frac{2}{(m-1)!}[ (m-2)!I + (m-3)!Q ] = \frac{2}{(m-1)}I + \frac{2}{(m-1)(m-2)}Q. 
\end{align*} This concludes the proof.

\end{proof}

\begin{proof}[\textbf{Proof of Theorem \ref{thm:recursivedesign}}]
As a base case, consider when $m = 4$. It is easy to verify that under Model (\ref{eqn:edgemodel}), the optimal fractional design $D_4$ has moment matrix \begin{align*}
    (4/4!)X^TX = (4/4!)\begin{bmatrix}
    6, 3, 3, 3 ,3 ,3 ,3 \\
    3, 3, 1, 1, 1, 1, 2 \\
    3, 1, 3, 1, 1, 2, 1 \\
    3, 1, 1, 3, 2, 1, 1 \\
    3, 1, 1, 2, 3, 1, 1 \\
    3, 1, 2, 1, 1, 3, 1 \\
    3, 2, 1, 1, 1, 1, 3
    \end{bmatrix} = (2/4!)\begin{bmatrix}
    12, 6, 6, 6, 6, 6, 6 \\
    6, 6, 2, 2, 2, 2, 4 \\
    6, 2, 6, 2, 2, 4, 2 \\
    6, 2, 2, 6, 4, 2, 2 \\
    6, 2, 2, 4, 6, 2, 2 \\
    6, 2, 4, 2, 2, 6, 2 \\
    6, 4, 2, 2, 2, 2, 6
    \end{bmatrix} 
\end{align*} By Theorem \ref{thm:fullmomentmat}, this matrix is equal to $M_f$ when $m = 4$. By Corollary \ref{cor:fulloptimal}, the fractional design must be $\phi-$ optimal, as it has the same moment matrix as the full design. 

\noindent \textbf{Induction Step.} Suppose that $D_{m-1}$ is an optimal fractional design for $m - 1$ components, with corresponding model matrix expansion $X_{m-1}$. It remains to show that if Algorithm \ref{Alg:getoptdesign} is used to generate the design for $m$ components, it also has moment matrix $M_f$.  

Since $X_{m-1}$ is $\phi-$optimal in $m-1$ components, and it has $\frac{(m-1)!}{4}$ rows, then by Theorem \ref{thm:fullmomentmat} it must satisfy the following properties: \begin{enumerate}
    \item $X_{m-1, ij}^TX_{m-1, ij} = (m-2)!/2$
    \item $X_{m-1, ij}^TX_{m-1, jk} = X_{m-1, ij}^TX_{m-1, ik} = (m-3)!/2$
    \item $X_{m-1, ij}^TX_{m-1, k\ell} = (m-3)!$
    \item $X_{m-1, 0}^TX_{m-1, 0} = (m-1)!/2$ and $X_{m-1, 0}^T X_{ij} = (m-2)!$
\end{enumerate} where $X_{m-1, ij}^TX_{m-1, k\ell}$ is the product of the columns of $X_{m-1}$ indexed by $ij, k\ell$, and $X_{m-1, 0}$ corresponds to the intercept column of $X_{m-1}$.

Consider an arbitrary undirected edge $(i,j)$ where $i < j$ for $i,j =  1,2,\dots,m-1$. Suppose that the pair are adjacent in columns $s-1, s$ of $D_{m-1}$. Then, by construction of the design in Algorithm \ref{Alg:getoptdesign}, the pair $ij$ will also be adjacent in columns $(s-1, s)$, $(s, s+1)$, $(m-s, m-s+1)$, or $(m-s-1,m-s)$ in all but one blocks of the blocks $B_1,\dots,B_m$ of $D_m$. Since there are $(m-2)!/2$ pairs of columns where this occurs, then it must follow that each undirected edge $(i, j)$ appears $(m-1)(m-2)!/2 = (m-1)!/2$ times in $D_m$, where $i, j = 1,2,\dots,m - 1$. 

Suppose a pair of undirected edges $(i,j)$, $(j, k)$ appear in consecutive columns $s-1, s, s+1$ in $D_{m-1}$. By construction of $D_m$, this pair of edges will appear consecutively in all but two of the blocks $B_1,\dots,B_m$. The argument for edges $(i, j), (i, k)$ is identical. Hence $X_{m, ij}^TX_{m, jk} = X^T_{m, ij}X_{m, ik} = (m-2)(m-3)!/2 = (m-2)!/2$ for $i,j,k = 1,2,\dots,m-1$. Similarly, the pair of undirected edges $(i,j), (k,\ell)$ will appear consecutively in all but two of the blocks of $B_1,\dots,B_m$. Hence $X_{m, ij}^TX_{m, k\ell} = (m-2)(m-3)! = (m-2)!$ for $i,j,k = 1,2,\dots,m-1$.

Consider an undirected edge $(i,m)$ where $i < m$. Since each row of $D_{m-1}$ is a valid permutation of $(1,2,\dots,m-1)$, then each vertex $i$ appears exactly once in each of the $\frac{(m-1)!}{4}$ rows of $D_{m-1}$. By construction, the undirected edge $(i,m)$ appears in exactly two blocks of $D_m$. It follows that $X_{m, im}^TX_{m, im} = \frac{2(m-1)!}{4} = (m-1)!/2$. 

Consider the case of consecutive undirected edges $(i,j), (j,m)$ with $i < j < m$. Each undirected edge $(i,j)$ appears with frequency $(m-3)!/2$ in $D_{m-1}$. By construction, the insertion of a column of $m$'s will separate the edge $(i,j)$ in exactly two of the blocks of $D_m$. Hence, these undirected edges are adjacent in all but two of the blocks $B_1,\dots,B_m$. It follows that $X_{m, ij}^TX_{m, jm} = (m-2)(m-3)!/2 = (m-2)!/2$. It remains to consider the general case of undirected edges $(i,j), (k,m)$ with $i<j<k<m$ appearing consecutively in $D_m$. The edges $(i,j),(j,k)$ appear with frequency $(m-3)!$ in $D_{m-1}$. By the same argument, inserting a column of $m$'s into each block of the algorithm will separate the edges $(i,j)$ and $(j,k)$ in exactly two blocks. It follows that $X_{m, ij}^TX_{m, km} = (m-2)(m-3)! = (m-2)!$.

Finally, the run size of the resulting design $D_m$ is $m(m-1)!/4 = m!/4$, so $X_{m, 0}^TX_{m, 0} = m!/4$, and by construction, the number of times each edge will appear in $D_m$ is $(m-1)(m-2)!/2 = (m-1)!/2$, so $X_{m, 0}^TX_{m, ij} = (m-1)!/2$. In summary, we have that:
\begin{enumerate}
    \item $X_{m, ij}^TX_{m, ij} = (m-1)!/2$
    \item $X_{m, ij}^TX_{m, jk} = (m-2)!$ and $X_{m, ij}^TX_{m,ik} = (m-2)!/2$
    \item $X_{m, ij}^TX_{m, k\ell} = (m-2)!$
    \item $X_{m, 0}^TX_{m, 0} = m!/4$ and $X_{m, 0}^TX_{m, ij} = (m-1)!/2$.
\end{enumerate}

Hence \begin{align*}
    \frac{4}{m!}X_m^TX_m = \frac{2}{m!} \begin{bmatrix}
    m!/2 & (m-1)!\mathbbm{1}^T \\
    (m-1)!\mathbbm{1} & (m-1)!I_{p} + (m-2)!Q 
    \end{bmatrix}    
    \end{align*}
    Since the fractional design has the same moment matrix as the full design, then it must be $\phi-$optimal. 
\end{proof}

\noindent \textbf{Proof of Theorem \ref{thm:minproblimit}.}
\begin{proof}
Let $\beta^* = [\beta_{12}^*, \dots, \beta_{m-1,m}^*]^T$ be the true value of $\beta$. Let $\hat{\beta}$ be the estimated model coefficients from model (\ref{eqn:edgemodel}). Notice that \begin{align*}
    \hat{\beta} &= (X^TX + R)^{-1}(X^TY + R\mu) \\
    &= (X^TX + R)^{-1}(X^T(X\beta^* + \epsilon) + R\mu) \\
    &= (X^TX + R)^{-1}(X^TX\beta^* + R\mu) + (X^TX + R)^{-1}(X^T\epsilon) \\
    &= (X^TX/n + R/n)^{-1}(X^TX\beta/n + R\mu/n) + (X^TX/n + R/n)^{-1}(X^T\epsilon/n)
\end{align*}

Let $\textbf{0}$ be a conformable zero matrix. Notice that $R/n \rightarrow \textbf{0}$ and $R\mu/n \rightarrow \textbf{0}$ as $n \rightarrow \infty$. By assumption (a) of Theorem \ref{thm:minproblimit}, it follows that \begin{align*}
    \hat{\beta} \rightarrow_p (Q + \textbf{0})^{-1}(Q\beta^* + \textbf{0}) + Q \cdot \textbf{0} \\
    \hat{\beta} \rightarrow_p Q^{-1}Q\beta^* = \beta^*.
\end{align*} So, $\hat{\beta} \rightarrow_p \beta^*$. It then follows that, for any $\textbf{a} \in \mathcal{A}^*$ \begin{align*}
    \hat{y}(\textbf{a}) = x(\textbf{a})^T\hat{\beta} \rightarrow_p x(\textbf{a})^T\beta^* = \tau(\textbf{a})
\end{align*}

Since this holds for all $\textbf{a} \in \mathcal{A}^*$, it must also hold for $\textbf{a}^*$ and $\hat{\textbf{a}}$. Therefore, we know that $\hat{y}(\hat{\textbf{a}}) \rightarrow_p \tau(\hat{\textbf{a}})$. Let $G$ be a graph on $m$ vertices with edge weights equal to $\beta_{jk}^*$ for $1 \leq j, k \leq m - 1$, and edge weights of 0 for any edge incident to vertex $m$. Then, for sufficiently large $n$, $\tau(\hat{\textbf{a}})$ is the output of a nearest neighbor algorithm on the graph $G$. Then by Assumption (b) and Theorem 1 of \cite{rosenkrantz1974approximate}, $\tau(\hat{\textbf{a}})/\tau(\textbf{a}^*) \leq 0.5\lceil\log_2(m)\rceil + 0.5$, so $\tau(\hat{\textbf{a}}) \leq C\tau(\textbf{a}^*)$ where $C = 0.5\lceil\log_2(m)\rceil + 0.5$. 

If $\tau(\hat{\textbf{a}}) = C\tau(\textbf{a}^*),$ then since $\hat{y}(\hat{\textbf{a}}) \rightarrow_P \tau(\hat{\textbf{a}})$, we have $\hat{y}(\hat{\textbf{a}}) \rightarrow_P C\tau(\textbf{a}^*)$ and the result follows. 

Suppose now that $\tau(\hat{\textbf{a}}) < C\tau(\textbf{a}^*)$. Let $\epsilon$ be a value in the open interval $(0, C\tau(\textbf{a}^*) - \tau(\hat{\textbf{a}}))$. Notice that if $|\hat{y}(\hat{\textbf{a}}) - \tau(\hat{\textbf{a}})| < \epsilon$, then \begin{align*}
    \hat{y}(\textbf{a}) < \tau(\hat{\textbf{a}}) + \epsilon < \tau(\hat{\textbf{a}}) + C\tau(\textbf{a}^*) - \tau(\hat{\textbf{a}}) = C\tau(\textbf{a}^*).
\end{align*} This implies that \begin{align}
    \label{ineq:squeeze1}
    P(|\hat{y}(\hat{\textbf{a}}) - \tau(\hat{\textbf{a}})| < \epsilon) \leq P(\hat{y}(\hat{\textbf{a}}) < C\tau(\textbf{a}^*)) \leq 1.
\end{align} Since $\hat{y}(\hat{\textbf{a}}) \rightarrow_p \tau(\hat{\textbf{a}})$ as $n \rightarrow \infty$, then it follows that the leftmost term in (\ref{ineq:squeeze1}) converges to 1 as $n \rightarrow \infty$. Taking the limit of all terms in (\ref{ineq:squeeze1}) yields \begin{align*}
    1 \leq \lim_{n \rightarrow \infty}  P(\hat{y}(\hat{\textbf{a}}) < C\tau(\textbf{a}^*)) \leq 1
\end{align*} which proves that $\lim_{n \rightarrow \infty}  P(\hat{y}(\hat{\textbf{a}}) < C\tau(\textbf{a}^*)) = 1$ for this case. 

\end{proof}

\end{document}